\documentclass[a4paper]{article} \usepackage{amstext}
\usepackage{subfigure} \usepackage{xspace}
\usepackage{amsthm} \textwidth 16.0cm\textheight 21.94cm
\newcommand{\D}{\displaystyle} 
\newcommand{\N}{\text{N}\xspace} \newcommand{\df}[1]{{\em {#1}}}
\newcommand{\urls}[1]{{\small \url{#1}}} \newcommand{\K}{K}
\newcommand{\X}{X} \newcommand{\V}{V} \newcommand{\+}{\oplus}
\newcommand{\prob}[4]{\begin{centering}\medskip\fbox{\begin{tabular}{l}{#1}\\
\begin{tabular}{lp{0.81\textwidth}}Input:
&{#2}\\Output: &{#3}\\Criterion:
&{#4}.\end{tabular}\end{tabular}}\medskip\end{centering}}

 \date{}
\usepackage{tikz} \usetikzlibrary{positioning}
\usetikzlibrary{shapes.misc} \usetikzlibrary{decorations.pathreplacing}
\usepackage[final,pdfpagelabels=false,colorlinks,linkcolor=blue,citecolor=blue,urlcolor=blue]{hyperref} \hypersetup{pdftitle={Bounds on series-parallel slowdown},pdfauthor={Andr\'as Z. Salamon, Vashti Galpin}}
\newtheorem{theorem}{Theorem}

\newtheorem{lemma}{Lemma}
\newtheorem{proposition}{Proposition}
\newtheorem{definition}{Definition}

\newtheorem{conjecture}{Conjecture}

\title{Bounds on series-parallel slowdown}
\author{Andr\'as Z. Salamon\\
Computing Laboratory, University of Oxford\\
Oxford-Man Institute of Quantitative Finance\\
\texttt{Andras.Salamon@comlab.ox.ac.uk}
\and Vashti Galpin\\
LFCS, School of Informatics, University of Edinburgh\\
\texttt{Vashti.Galpin@ed.ac.uk}
}

\begin{document}

\maketitle

\begin{abstract}
We use activity networks (task graphs) to model
parallel programs and consider series-parallel extensions of these
networks. Our motivation is two-fold: the benefits of series-parallel
activity networks and the modelling of programming constructs,
such as those imposed by current parallel computing environments.
Series-parallelisation adds precedence constraints to an activity
network, usually increasing its makespan (execution time). The slowdown
ratio describes how additional constraints affect the makespan.
We disprove an existing conjecture positing a bound of two on the
slowdown when workload is not considered.  Where workload is known,
we conjecture that 4/3 slowdown is always achievable, and prove our
conjecture for small networks using max-plus algebra. We analyse a
polynomial-time algorithm showing that achieving 4/3 slowdown is in
exp-APX. Finally, we discuss the implications of our results.
\end{abstract}

\section{Introduction}

An approach to reducing the execution time of a computer program is to
run it on multiple processors simultaneously.  The study of parallel
programming and architectures has seen a resurgence with the widespread
adoption of multi-core processing units in computing systems.
Commercial numerical software such as MATLAB\footnote{Via the MATLAB
Parallel Computing Toolbox. \urls{http://www.mathworks.com/}} and
Mathematica\footnote{From version 7. \urls{http://www.wolfram.co.uk/}}
can now take advantage of multiple processors, and OpenCL is a
recently finalised standard for programming with multiple-processor
systems \cite{Munshi2009:opencl}.

An important aspect of parallel programming is scheduling, the method
by which code is allocated to processors \cite{Kwok1999:static}.
Here we instead consider the inherent precedence constraints of a
parallel program and the constraints imposed by tranformation and by
the programming constructs that are used to describe parallelism,
both of which affect execution time. Our concerns are orthogonal
to scheduling since we assume sufficient processors and hence the
decision on what to schedule next is unimportant.

A program can be divided up into activities or tasks.  This can be
done in different ways depending on the granularity used.  Here we do
not consider granularity further but assume some reasonable approach
has been used.  The activities can be related to each other by the
order in which they must occur for the program to work correctly.
For instance, if one activity modifies a variable and another activity
uses this modified value, then the modifying activity must occur before
the activity uses the new value.  An activity that must occur before
another \df{precedes} the other activity and there is a \df{precedence
constraint} between the two activities.  Precedence is imposed by
the structure of the program and is inherent to the particular set
of activities.

The formalism used to describe precedences between activities is
known as an \df{activity network}, \df{network} or \df{task graph}.
We use the \df{activity-on-node} variant of this model, where weights
are associated with the vertices of the network.  These and variants
such as PERT networks or machine schedules (sometimes with edge instead
of node weights) are widely used in fields such as project management,
operational research, combinatorial optimization and computer science.

Activity networks can be classified by their structure.  Structures
of interest are series-parallel (SP), and level-constrained (LC)
\cite{malony:94} which are a proper subset of SP and a subset of of
Bulk Synchronous Programming (BSP) which has been used successfully as
an approach to parallel programming \cite{bisseling2004,valiant90}.
Analysis of activity networks is difficult but is easier
for SP \cite{dodin85}. For instance, scheduling is NP-hard but
polynomial-time for SP networks \cite{Garey1979:computers}.  We call
the addition of constraints to achieve an SP activity network
\df{series-parallelisation} (SP) \cite{Gonzalez2002:mapping}.

Programming constructs can also impose an SP structure over and
above the inherent constraints. The most obvious is the sequencing
of commands in a sequential programming language but the addition
of constraints can also occur with parallel constructs as we show in
the motivating example in Section~\ref{eg}.

The precedence constraints between activities determine the minimum
time to execute the program.  Assuming a sufficient number of
processors and non-preemptive, work-conserving scheduling the fastest
time for execution will be the time taken to execute slowest chains of
activities, called \df{critical paths}. Chains consist of activities
that are totally ordered and hence must proceed one after another,
excluding the possibility of parallelism.

This paper considers the difference in execution time between activity
networks, comparing a network with only inherent precedence constraints
with the same network with added precedence constraints that make it
an SP structure.  Adding constraints results in programs that take at
least as long and we consider the slowdown where slowdown is the ratio
of the slower program to the faster one.  We characterise the slowdown
induced by LC and disprove an existing conjecture about slowdown for
SP \cite{Vangemund1997:importance}. This requires demonstrating that
large slowdown can occur for every possible series-parallelisation of
a specific network.  A new conjecture is presented, and results proved
for small instances.  Additionally we discuss the complexity of finding
the optimal SP for a network.  First we present a motivating example,
followed by background and definitions of the relevant structures
after which come the main results and conjecture. We finish with the
implications of our results and further research.

\section{Motivating example}\label{eg}

\begin{figure}[t] \centering \mbox{
\subfigure[Neighbour synchronisation example]{ \begin{tikzpicture}[node
distance=0.8cm,xscale=1.2,yscale=-0.75,>=stealth]
\tikzstyle{v}=[circle,minimum size=1mm,inner sep=0pt,draw]
\tikzstyle{vn}=[circle,minimum size=1mm,inner sep=0pt] \node at
(-1,0) {}; \node at (3,0) {}; \node[v] (a11) at (0,0) {}; \node
[right=+4pt of a11.east] {$a_{1,1}$}; \node[v] (a12) at (1,0) {};
\node [right=+2pt of a12.east] {$a_{1,2}$}; \node[v] (a13) at (2,0)
{}; \node [right=+2pt of a13.east] {$a_{1,3}$}; \node (a14) at (3,0)
{\ldots}; \node[v] (a1N) at (4,0) {}; \node [right=-1pt of a1N.east]
{$a_{1,m}$}; \node[v] (a21) at (0,1) {}; \node [right=+4pt of a21.east]
{$a_{2,1}$}; \node[v] (a22) at (1,1) {}; \node [right=+2pt of a22.east]
{$a_{2,2}$}; \node[v] (a23) at (2,1) {}; \node [right=+2pt of a23.east]
{$a_{2,3}$}; \node (a24) at (3,1) {\ldots}; \node[v] (a2N) at (4,1)
{}; \node [right=-1pt of a2N.east] {$a_{2,m}$}; \node[v] (a31) at (0,2)
{}; \node[v] (a32) at (1,2) {}; \node[v] (a33) at (2,2) {}; \node[vn]
(a34) at (3,2) {}; \node (a34n) [right=of a33] {\ldots}; \node[v]
(a3N) at (4,2) {}; \node [below=-2.3pt of a31.south] {$\vdots$};
\node [below=-2.3pt of a32.south] {$\vdots$}; \node [below=-2.3pt of
a33.south] {$\vdots$}; \node [below=-2.3pt of a3N.south] {$\vdots$};
\node[v] (a311) at (0,3) {}; \node[v] (a312) at (1,3) {}; \node[v]
(a313) at (2,3) {}; \node[vn] (a314) at (3,3) {}; \node (a314n)
[right=of a313] {\ldots}; \node[v] (a31N) at (4,3) {}; \node[v] (at1)
at (0,4) {}; \node [right=+4pt of at1.east] {$a_{t,1}$}; \node[v] (at2)
at (1,4) {}; \node [right=+2pt of at2.east] {$a_{t,2}$}; \node[v]
(at3) at (2,4) {}; \node [right=+2pt of at3.east] {$a_{t,3}$};
\node (at4) at (3,4) {\ldots}; \node[v] (atN) at (4,4) {}; \node
[right=-1pt of atN.east] {$a_{t,m}$}; \draw[->,thick] (a11) to (a21);
\draw[->,thin,gray] (a12) to (a21); \draw[->,thick] (a11) to (a22);
\draw[->,thin,gray] (a12) to (a22); \draw[->,thick] (a13) to (a22);
\draw[->,thin,gray] (a12) to (a23); \draw[->,thin,gray] (a13) to (a23);
\draw[->,thin,gray] (a1N) to (a2N); \draw[->,thin,gray] (a21) to (a31);
\draw[->,thin,gray] (a22) to (a31); \draw[->,thin,gray] (a21) to (a32);
\draw[->,thin,gray] (a22) to (a32); \draw[->,thin,gray] (a23) to (a32);
\draw[->,thin,gray] (a22) to (a33); \draw[->,thin,gray] (a23) to (a33);
\draw[->,thin,gray] (a2N) to (a3N); \draw[->,thin,gray] (a311) to
(at1); \draw[->,thin,gray] (a312) to (at1); \draw[->,thin,gray]
(a311) to (at2); \draw[->,thin,gray] (a312) to (at2);
\draw[->,thin,gray] (a313) to (at2); \draw[->,thin,gray] (a312)
to (at3); \draw[->,thin,gray] (a313) to (at3); \draw[->,thin,gray]
(a31N) to (atN); \draw[dashed,->] (a11) to (a23); \draw[dashed,->]
(a11) to (a2N); \draw[dashed,->] (a12) to (a2N); \draw[dashed,->]
(a13) to (a21); \draw[dashed,->] (a13) to (a2N); \draw[dashed,->]
(a21) to (a33); \draw[dashed,->] (a21) to (a3N); \draw[dashed,->]
(a22) to (a3N); \draw[dashed,->] (a23) to (a31); \draw[dashed,->]
(a23) to (a3N); \draw[dashed,->] (a311) to (at3); \draw[dashed,->]
(a311) to (atN); \draw[dashed,->] (a312) to (atN); \draw[dashed,->]
(a313) to (at1); \draw[dashed,->] (a313) to (atN); \end{tikzpicture}
\label{neighboursyn} } \qquad \qquad \subfigure[\N network]{
\begin{tikzpicture}[xscale=1.75,yscale=-1.66,>=stealth] \node at (1,0)
{}; \node at (2,2.5) {}; \tikzstyle{v}=[circle,minimum size=1mm,inner
sep=0pt,draw] \node[v] (a) at (1,1) {}; \node[v] (b) at (2,1) {};
\node[v] (c) at (1,2) {}; \node[v] (d) at (2,2) {}; \draw[->,thin,gray]
(a) -- (c); \draw[->,thin,gray] (a) -- (d); \draw[->,thin,gray] (b) --
(d); \end{tikzpicture} \label{N} } } \caption{Activity
networks} \end{figure}
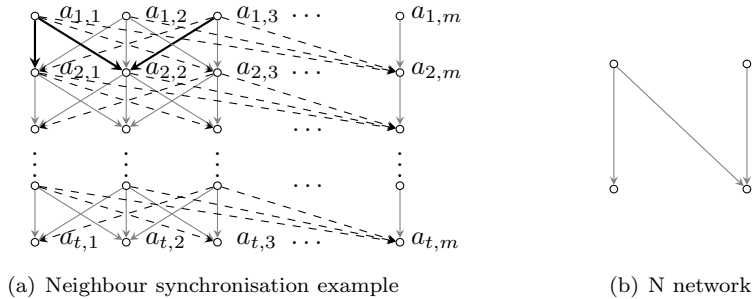

We next consider a simple example involving computations dependent
on earlier computations. In a 1-dimensional flow model of heat
diffusion in a metal rod, we calculate the temperature at $m$ points
for each time step. The temperature at time $\tau+1$ at point $p_i$ is
dependent on the temperature at time $\tau$ at points $p_{i-1}$, $p_i$
and $p_{i+1}$. If we view each calculation as an activity $a_{i,\tau}$,
this is an example of neighbour synchronisation (NS) as illustrated
in Figure~\ref{neighboursyn} when considering the solid lines only.
This network is not SP because of the edges $(a_{1,1},a_{2,1})$,
$(a_{1,1},a_{2,2})$ and $(a_{1,3},a_{2,2})$ and the lack of the
edge $(a_{1,3},a_{2,1})$. This is an example of the smallest non-SP
activity network, the \N network shown in Figure~\ref{N}.  There are
many instances of \N in the example activity network.

An obvious (although not necessarily the best) way to
series-parallelise this activity network is to require all activities
at time $\tau$ to precede those at time $\tau+1$.  The dashed lines in
Figure~\ref{neighboursyn} illustrate the added precedence constraints.
The edge $(a_{1,3},a_{2,1})$ is added as well as edges to remove
the other \N networks.  Figure~\ref{neighboursyn} is an example of
a level-constrained (LC) extension.

Assume unit workloads for all activities apart from one much
slower activity at each time instance $\tau$ with duration
$t(a_{\tau,2\tau-1})=C \gg 1$. Hence for every calculation of
a specific point over time, there is only one large workload.
The execution time for the above series-parallelisation will
be $(C-1)(m+1)/2 + s$ where $s\geq n$ is the total number of
timesteps. This gives large slowdown since the execution time
considering only inherent constraints is $C+s-1$.

There may be better ways to series-parallelise this network, however
a language such as MATLAB may impose a particular SP activity
network through its programming constructs. If one expresses this
example as parallel code using the \texttt{parfor} statement (in
the obvious simple way) then one will achieve the SP network given
in Figure~\ref{neighboursyn}.

An understanding of the slowdown obtained by various forms of
series-parallel\-isation is therefore important, particularly due
to the increased usage of parallel programming constructs to take
advantage of multi-core processors.

\section{Background}

This section defines notation and basic concepts for activity-on-node
networks.

\begin{definition} An \emph{activity-on-node network} \emph{(task
graph, activity network}, or simply, \emph{network)} consists of
\begin{itemize} \item $V=\{a_1,\dots,a_n\}$ a set of \df{activities},
\item $G=(V,E)$ a directed acyclic graph with \df{precedence
constraints} $E \subseteq V \times V$, \item $t:V\rightarrow(0,\infty)$
a \df{workload} assigning a duration to each activity.  \end{itemize}
\end{definition} A precedence constraint $(a,b)$ captures the idea that
activity $a$ must complete before activity $b$ can begin.  We assume
that we are working with the transitive closure of the precedence
constraints, namely that the precedence relation is irreflexive and
transitive.  However, when drawing activity networks, we only draw
the edges that appear in the transitive reduction of the network.

The \df{makespan} of an activity network $G$, denoted $T(G)$, is the
time to complete all activities of the network.  This depends on the
scheduling policy and the number and configuration of processors.
We make the following assumptions.

\textbf{Scheduling:} We assume non-preemptive scheduling, namely
once an activity is assigned to a processor, it will complete on
that processor without interruption; and a work-conserving scheduling
policy, namely no processor is left idle if there are still activities
waiting to start.

\textbf{Number and type of processors:} The processors are identical
and there are sufficiently many, in the sense that any activity that
is ready to execute can be started.  It is sufficient to have as many
processors available as the width of the activity network.

\textbf{Overheads:} All overheads such as communication, contention
and decisions about which activity to execute next are included in
the workload.

\medskip

\noindent Given these assumptions, we can characterise the makespan
of activity networks.

\begin{definition} Let $G=(V,E)$ and $G'=(V',E')$ be directed graphs.
\begin{itemize} \item $G$ is a \emph{subgraph} of $G'$, $G \subseteq
G'$ if $V \subseteq V'$ and $E \subseteq E'$.  \item If $G$ is a
\emph{subgraph} of $G'$ then $G'$ is a \emph{supergraph} of $G$.
\item $G$, a {subgraph} of $G'$, is an \df{antichain} if $E$
is empty.  \item $G$, a {subgraph} of $G'$, is a \df{chain} if $E$
is a total order over $V$.  \item $G'$, a {supergraph} of $G$, is
an \df{extension} if $E \subseteq E'$ and $V=V'$.  \end{itemize}
\end{definition} An extension formally defines what it means to add
precedence constraints and does not permit addition of activities so
$t$ remains unchanged. A \df{subnetwork} has the obvious meaning.

\begin{definition} Let $G=(V,E)$ be an activity network.
\begin{itemize} \item $depth(G) = \max\{ \bigl| C \bigr| \mid C \text{\
a chain in\ } G \}. $ \item $ width(G) = \max\{ \bigl| A \bigr| \mid
A \text{\ an antichain in\ } G \}. $ \end{itemize} \end{definition}
A chain represents its activities occuring one after the other,
and hence the time taken for a chain to execute is the sum of the
durations for each activity.  \begin{proposition} The makespan of
a chain $C=(V,E)$ with $V=\{a_1,\ldots,a_n\}$ is \newline \indent
\indent $ T(C) = \sum_{i=1}^n t(a_i). $ \end{proposition} The makespan
of an activity network can be characterised as the time it takes to
complete a chain in the network with the longest completion time (a
critical path).  The proof is straightforward, and makes essential use
of the work-conserving property of the scheduling policy, and the fact
that there are sufficient processors.  If the number of processors
is insufficient, a work-conserving approach may be sub-optimal
\cite{Kohler1975:preliminary}.  \begin{proposition} The makespan of an
activity network $G=(V,E)$ is \newline \indent \indent $ T(G) = \max\{
T(C) \mid C \text{\ is a chain in\ } G \}. $ \end{proposition} When we
create extensions by adding constraints to obtain a specific network
structure, we cannot decrease the time that the activity network will
take to complete \cite{kleinoder82,Salamon2001:thesis}.  We can define
the ratio between the two makespans as a slowdown\footnote{If
we were comparing a sequential program with its parallel version,
we would consider \emph{speedup}, namely the ratio of the faster to
the slower. Since we know that the program with additional precedence
constraints will take at least as long as the original, we consider
slowdown, the ratio of the slower to the faster.}.

\begin{definition} Let $H$ be an extension of $G$ then the
\emph{slowdown} is $T(H)/T(G)$.  \end{definition}

\section{Structure of activity networks}\label{structure}

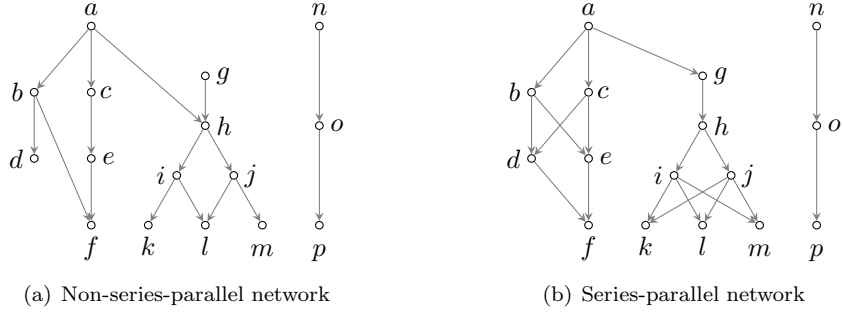
\begin{figure}[t] \centering
\mbox{ \subfigure[Non-series-parallel network]{
\begin{tikzpicture}[xscale=0.75,yscale=-0.66,>=stealth]
\tikzstyle{v}=[circle,minimum size=1mm,inner sep=0pt,draw] \node
at (0,1) {}; \node at (7,1) {}; \node[v] (u) at (2,1) {}; \node[v]
(v1) at (1,2.33) {}; \node[v] (v2) at (2,2.33) {}; \node[v] (w1) at
(1,3.66) {}; \node[v] (w2) at (2,3.66) {}; \node[v] (x) at (2,5) {};
\node[v] (r) at (4,2) {}; \node[v] (s) at (4,3) {}; \node[v] (t1) at
(3.5,4) {}; \node[v] (t2) at (4.5,4) {}; \node[v] (y1) at (3,5) {};
\node[v] (y2) at (4,5) {}; \node[v] (y3) at (5,5) {}; \node[v] (z1)
at (6,1) {}; \node[v] (z2) at (6,3) {}; \node[v] (z3) at (6,5) {};
\node [above=-1pt of u.north] {$a$}; \node [left=-1pt of v1.west]
{$b$}; \node [right=-2pt of v2.east] {$c$}; \node [left=-1pt
of w1.west] {$d$}; \node [right=-1pt of w2.east] {$e$}; \node
[below=-1pt of x.south] {$f$}; \node [right=-1pt of r.east] {$g$};
\node [right=-1pt of s.east] {$h$}; \node [left=-1pt of t1.west]
{$i$}; \node [right=-1pt of t2.east] {$j$}; \node [below=-1pt
of y1.south] {$k$}; \node [below=-1pt of y2.south] {$l$}; \node
[below=+2pt of y3.south] {$m$}; \node [above=-1pt of z1.north] {$n$};
\node [right=-1pt of z2.east] {$o$}; \node [below=+2pt of z3.south]
{$p$}; \draw[->,thin,gray] (u) -- (v1); \draw[->,thin,gray] (u) --
(v2); \draw[->,thin,gray] (v1) -- (w1); \draw[->,thin,gray] (v2) --
(w2); \draw[->,thin,gray] (v1) -- (x); \draw[->,thin,gray] (w2) --
(x); \draw[->,thin,gray] (u) -- (s); \draw[->,thin,gray] (r) -- (s);
\draw[->,thin,gray] (s) -- (t1); \draw[->,thin,gray] (s) -- (t2);
\draw[->,thin,gray] (t1) -- (y1); \draw[->,thin,gray] (t1) -- (y2);
\draw[->,thin,gray] (t2) -- (y2); \draw[->,thin,gray] (t2) -- (y3);
\draw[->,thin,gray] (z1) -- (z2); \draw[->,thin,gray] (z2) -- (z3);
\end{tikzpicture} \label{nsp2} } \qquad \subfigure[Series-parallel
network]{ \begin{tikzpicture}[xscale=0.75,yscale=-0.66,>=stealth]
\node at (0,1) {}; \node at (7,1) {}; \tikzstyle{v}=[circle,minimum
size=1mm,inner sep=0pt,draw] \node[v] (u) at (2,1) {}; \node[v] (v1) at
(1,2.33) {}; \node[v] (v2) at (2,2.33) {}; \node[v] (w1) at (1,3.66)
{}; \node[v] (w2) at (2,3.66) {}; \node[v] (x) at (2,5) {}; \node[v]
(r) at (4,2) {}; \node[v] (s) at (4,3) {}; \node[v] (t1) at (3.5,4)
{}; \node[v] (t2) at (4.5,4) {}; \node[v] (y1) at (3,5) {}; \node[v]
(y2) at (4,5) {}; \node[v] (y3) at (5,5) {}; \node[v] (z1) at (6,1)
{}; \node[v] (z2) at (6,3) {}; \node[v] (z3) at (6,5) {}; \node
[above=-1pt of u.north] {$a$}; \node [left=-1pt of v1.west] {$b$};
\node [right=-2pt of v2.east] {$c$}; \node [left=-1pt of w1.west]
{$d$}; \node [right=-1pt of w2.east] {$e$}; \node [below=-1pt of
x.south] {$f$}; \node [right=-1pt of r.east] {$g$}; \node [right=-1pt
of s.east] {$h$}; \node [left=-1pt of t1.west] {$i$}; \node [right=-1pt
of t2.east] {$j$}; \node [below=-1pt of y1.south] {$k$}; \node
[below=-1pt of y2.south] {$l$}; \node [below=+2pt of y3.south] {$m$};
\node [above=-1pt of z1.north] {$n$}; \node [right=-1pt of z2.east]
{$o$}; \node [below=+2pt of z3.south] {$p$}; \draw[->,thin,gray]
(u) -- (v1); \draw[->,thin,gray] (u) -- (v2); \draw[->,thin,gray]
(v1) -- (w1); \draw[->,thin,gray] (v1) -- (w2); \draw[->,thin,gray]
(v2) -- (w1); \draw[->,thin,gray] (v2) -- (w2); \draw[->,thin,gray]
(w1) -- (x); \draw[->,thin,gray] (w2) -- (x); \draw[->,thin,gray] (u)
-- (r); \draw[->,thin,gray] (r) -- (s); \draw[->,thin,gray] (s) --
(t1); \draw[->,thin,gray] (s) -- (t2); \draw[->,thin,gray] (t1) --
(y1); \draw[->,thin,gray] (t1) -- (y2); \draw[->,thin,gray] (t1) --
(y3); \draw[->,thin,gray] (t2) -- (y1); \draw[->,thin,gray] (t2) --
(y2); \draw[->,thin,gray] (t2) -- (y3); \draw[->,thin,gray] (z1) --
(z2); \draw[->,thin,gray] (z2) -- (z3); \end{tikzpicture} \label{sp1}
} } \caption{Series-parallelisation of an activity
network} \end{figure}

We need to define what it means for a activity network to be
series-parallel.  Figure~\ref{nsp2} is not SP and Figure~\ref{sp1}
is SP. The \N network in Figure~\ref{N} is also not SP.
An activity network is SP if it consists of a single activity
or can be recursively decomposed into chains and antichains
using series and parallel composition.  \begin{definition} An
activity network $G=(V,E)$ is \emph{series-parallel (SP)} if $G$
can be expressed using the SP grammar $g ::= (g\+g) \mid g\cdot
g \mid a$ where $a$ is an activity, and each activity appears at
most once.  A string generated by the SP grammar is an \emph{SP
expression}.  \end{definition} We also use juxtaposition $G_1G_2$ for
$G_1\cdot G_2$. The network in Figure~\ref{sp1} can be expressed as
$(a(((b\+c)(d\+e)f)\+(gh(i\+j)(k\+l\+m)))\+nop)$.  \begin{definition}
Let $G_1 = (V_1,E_1)$ and $G_2 = (V_2,E_2)$ be activity networks
with $V_1 \cap V_2 = \emptyset$ and $E_1 \cap E_2 = \emptyset$.
\begin{itemize} \item The \emph{parallel composition} of $G_1$
and $G_2$ is $G_1\+G_2 = (V_1\cup V_2,E_1\cup E_2)$.  \item The
\emph{series composition} of $G_1$ and $G_2$ is $G_1\cdot G_2 =
(V_1\cup V_2,(V_1 \times V_2) \cup E_1\cup E_2)$.  \end{itemize}
\end{definition} SP networks are exactly those that do not contain the
\N network \cite{Valdes1982:recognition}.  If we have a network that
is not SP, we can add constraints until it is SP. An SP extension
of an activity network always exists since if we add sufficient
constraints we obtain a chain, which is SP \cite{fishburn85}.
The activity network in Figure~\ref{sp1} is a series-parallelisation
of the activity network in Figure~\ref{nsp2}. We can easily calculate
the makespan of an SP network.

\begin{proposition} Let $G=(V,E)$ be an SP activity network with
SP expression $g$.  The makespan of $G$ is $T(G)=T(g)$ where
\newline \hspace*{0.2cm} $T((g_1\+g_2)) = \max\{T(g_1),T(g_2)\},
\quad T(g_1\cdot g_2) = T(g_1)+T(g_2), \quad T(a) = t(a)$.
\end{proposition} This links the SP grammar with the max-plus algebra
\cite{cuninghame79}.  For convenience, the symbol $\+$ will denote
$\max$ and the symbol $\cdot $ will denote arithmetic $+$.

Level-constrained networks are a strict subset of SP.  The level of
an activity $a$ is the size of a maximal chain in the network which
has $a$ as its last activity.

\begin{definition} For an activity network $G=(V,E)$, the \df{level}
of an activity $a$ is \newline \indent $ \lambda(a) = \max \{ \bigl|
C \bigr| \mid C \text{\ is a chain in $G$, $a \in C$, and for all\ }
b \in C, (b,a) \in E \}.  $ \end{definition} The level of each activity
in a network can be computed in polynomial time, by marking activities
in a breadth-first search of the network's transitive reduction.  The
depth of an activity network is the maximum level of its activities.
We can now add precedence constraints to obtain an extension of
the network that maintains its level structure.  This is a common
technique \cite{malony:94,Salamon2001:thesis}.  \begin{definition}
For an activity network $G=(V,E)$, the \df{level-constrained (LC)
extension} of $G$ is the network $G_L=(V,E_L)$, where $ E_L = \{ (a,b)
\mid \lambda(a) < \lambda(b) \}.  $ \end{definition} Note that $G_L$
is an extension of $G$, and that $depth(G) = depth(G_L)$.  We can
identify a level as $\Lambda_i = \{a \in V \mid \lambda(a) = i \}$;
each level is an antichain and the levels partition the activities.
$G_L$ is also in BSP form \cite{valiant90} since each level consists
of independent chains (of size one, in this case) and all activities
in one level must complete before any activity in the next level
can start. LC networks have the form $\alpha_1\ldots\alpha_d$ where
$\alpha_i = (a_{i,1}\+\ldots\+a_{i,m_i})$.

We consider a structure that is non-SP for sufficiently large networks.

\begin{definition} A \df{neighbour synchronisation (NS) network} of
depth $d$, width $w$, and degree $\Delta$, denoted $ns(d,w,\Delta)$,
consists of activities $a_{i,j}$ with $i \in \{1,\ldots,d\}$, $j \in
\{1,\ldots,w\}$, and precedence constraints $(a_{i,j},a_{i+1,j+k})$ for
every $k = -\lfloor (\Delta-1)/2\rfloor , -\lfloor (\Delta-1)/2\rfloor
+1, \ldots, \lceil (\Delta-1)/2 \rceil $ (as long as $1 \le j+k
\le w$).  \end{definition} Figure~\ref{neighboursyn} depicts an NS
network of depth $t$, width $m$, and degree $3$.  The dashed precedence
constraints are those added by the process of LC extension.

\section{Bounding LC slowdown}\label{lcbound}

There are three reasons for considering LC networks.  First, they
relate to BSP, a useful applied technique for parallel programming,
and second, they are efficient to construct for any given activity
network.  Last, it is straightforward to construct an upper-bound on
the slowdown for a given workload $t$.

\begin{theorem}\label{thm:rho} Given an activity network $G=(V,E)$
and its LC extension $G_L=(V,E_L)$, the slowdown is bounded by
the ratio $\rho$ of the largest to the smallest duration in the
workload.  \[ \frac{T(G_L)}{T(G)} \leq \frac{\max\{t(a)\mid a \in
V\}}{\min\{t(a)\mid a \in V\}} = \rho.  \] \end{theorem} \begin{proof}
Given an LC extension $G_L = (V,E_L)$ of $G$, its makespan is $ T(G_L)
\: = \: \sum_{i=1}^{depth(G)} \max\{ t(a) \mid a \in \Lambda_i \} \:
\leq \: depth(G).\max\{ t(a) \mid a \in V\} $ since a critical path
is determined by the slowest activity at each level and is bounded
by the depth times the largest duration. Also \( depth(G).\min\{
t(a) \mid a \in V\} \le T(G) \) since the depth of $G$ is the size
of the longest chain and the time taken for each activity in this
chain is at least as long as the activity with the shortest duration.
The result follows from these two inequalities. \end{proof}

By Theorem~\ref{thm:rho}, if all activities have similar durations,
then the slowdown will be close to one.  If we know in advance that
$\rho$ is small, then it is reasonable to series-parallelise using an
LC extension.  This is efficient to obtain, and BSP is then also an
appropriate model for the computation, since any BSP can be transformed
to LC by treating independent chains as single activities.

Conversely, if $\rho$ is large, its importance depends on how tight it
is. If it is tight, and we know that large values may occur because
an activity could be delayed (for instance, due to a cache miss, or
swapping to and from disk, or because of competition for resources),
then the large value of $\rho$ indicates that a LC extension is a
poor choice for series-parallelization.

By considering $ns(d,w,3)$ with $w \ge 2d-1$, we can demonstrate that
slowdown for the LC extension can be arbitrarily close to $\rho$.
\begin{proposition} For any $\epsilon > 0$, there exists an NS
activity network $G$ and a workload $t$ such that $\rho - T(G_L)/T(G)
< \epsilon$.  \end{proposition} However, $\rho$ can be pessimistic:
consider $ns(1,w,3)$ with one large activity and many small ones.
This is already SP, yet $\rho$ can be made arbitrarily large.

The next section presents two conjectures about bounds for general
series-parallelisations of activity networks.

\section{Bounding SP slowdown}

This section considers a conjecture by van Gemund
\cite{Vangemund1997:importance}.  We need to introduce a parameterised
notation for makespan. Denote the makespan by $T(G,t)$ to indicate
specifically the role of the workload function $t$.  There are
two different classes of algorithms that can be used to obtain a
series-parallelisation.  We use the notation $S(G,t)$ to denote the
SP network that is the output of some algorithm that considers both
the graph and the workload, and $S'(G)$ to denote the SP network
that is the output of some algorithm that considers only the graph.
Using this notation we can posit two distinct hypotheses:

{\renewcommand{\arraystretch}{1.4} $\begin{array}{lp{0.1cm}lllllll}
\textbf{Workload-independent:} & & \D \exists \kappa\: & \forall G\:
& \exists S' & \forall t\: & \bigl[T(S'(G),t)/T(G,t) \leq \kappa\bigr]
\\ \textbf{Workload-dependent:} & & \D \exists \kappa\: & \forall G\:
& \forall t & \exists S\: & \bigl[T(S(G,t),t)/T(G,t) \leq \kappa\bigr]
\\ \end{array}$}

\noindent These can be understood as follows.  The first states that
for every graph, there is a series-parallelisation with a slowdown
bound of $\kappa$ that works for every possible workload on that graph
and the second states that given a graph and a workload, there is a
series-parallelisation with slowdown bound of $\kappa$.

Van Gemund \cite{Vangemund1997:importance} conjectures that $\kappa =
2$ is a bound for slowdown for the workload-independent case.

\begin{conjecture}[\cite{Vangemund1997:importance}] \label{factortwo}
For any activity network $G=(V,E)$, it is possible to find a SP
extension $G_{SP}$ of $G$, such that for every workload $t \colon
V \rightarrow (0,\infty)$, \[ \frac{T(G_{SP},t)}{T(G,t)} \leq 2.
\] \end{conjecture} There is an algorithm that meets this bound
under ``reasonable'' workloads \cite{Vangemund1997:importance}.
The following result disproves Conjecture~\ref{factortwo}.
\begin{theorem}\label{nofactortwo} For any series-parallelisation of
$Q = ns(3,8,3)$, there exists a workload leading to slowdown greater
than $2$.  \end{theorem} We need some lemmas for the proof.

\begin{lemma}\label{connected} Any SP extension of a weakly connected
network $G$ will have an SP expression of the form $\alpha \beta$,
where both $\alpha$ and $\beta$ are SP expressions.  \end{lemma}

\begin{proof} An SP expression $\alpha\+\beta$ has no constraints
between activities in $\alpha$ and in $\beta$, so the network is
disconnected. The result follows by contradiction. \end{proof}

\begin{lemma}\label{updownclosed} Suppose an NS network $G$ has
SP expression $\alpha\beta$ with $\Delta$ odd.  \begin{enumerate}
\item If $a_{i,j}$ is in $\alpha$ then $a_{k,l}$ is also, whenever
$(\Delta-1)(i-k)/2 \ge |j-l|$.  \item If $a_{i,j}$ is in $\beta$
then $a_{k,l}$ is also, whenever $(\Delta-1)(k-i)/2 \ge |j-l|$.
\end{enumerate} \end{lemma}

\begin{proof} Suppose $a_{i,j}$ is in $\alpha$.  By the definition
of NS networks, $(\Delta-1)(i-k)/2 \ge |j-l|$ means that $a_{k,l}$
precedes $a_{i,j}$.  If $a_{k,l}$ were in $\beta$ then $a_{i,j}$
would precede $a_{k,l}$, which is impossible.  The second part is
symmetric. \end{proof} For a network $G$ and an SP expression
$\alpha$, let $G|_\alpha$ denote the subnetwork of $G$ consisting of
only those activities that appear in $\alpha$.

\begin{lemma}\label{depth} Suppose $d \ge 3$ and $w \ge 3$. Any SP
extension of $ns(d,w,3)$ will have an SP expression of the form $\alpha
\beta$, where either $G|_\alpha$ or $G|_\beta$ is not SP.  \end{lemma}

\begin{proof} We argue a contradiction for $ns(3,3,3)$; the result
follows for larger $w$ and $d$ by considering any subnetwork
isomorphic to $ns(3,3,3)$ which is not completely contained in
either $\alpha$ or $\beta$.  Suppose $\alpha$ and $\beta$ are
both SP.  Suppose activity $a_{2,1}$ and $a_{2,3}$, are both in
$\alpha$ without loss of generality.  By Lemma~\ref{updownclosed},
$a_{1,1}$ and $a_{1,2}$ are then both in $\alpha$ or both in
$\beta$.  However, $\{a_{1,1},a_{1,2},a_{2,1},a_{2,3}\}$ forms
an \N network in $G$, so $G|_\alpha$ cannot be SP.  Now suppose
activity $a_{2,1}$ is in $\alpha$ and $a_{2,3}$ is in $\beta$ (the
opposite arrangement is symmetric).  If $a_{2,2}$ is in $\alpha$ then
$\{a_{1,1},a_{1,3},a_{2,1},a_{2,2}\}$ forms an \N network in $G$; if
$a_{2,2}$ is in $\beta$ then $\{a_{2,2},a_{2,3},a_{3,1},a_{3,3}\}$
forms an \N network in $G$.  Hence at least one of $G|_\alpha$
or $G|_\beta$ is not SP. \end{proof} The depth of $3$ in
Lemma~\ref{depth} is necessary, as any NS network of depth 2 can be
made SP by enforcing level 1 to precede level 2, and each level is
an SP network.  Further, any width $2$ NS network is SP, so the width
of $3$ is also necessary.

\begin{proof}[Proof (of Theorem~\ref{nofactortwo})] We show that in any SP
extension $Q'$ of $Q$, there must exist three activities $a,b,c$ which
form an antichain in $Q$ but a chain in $Q'$, and then construct a
suitable workload using this chain.  Possible arrangements of $a,b,c$
are illustrated.

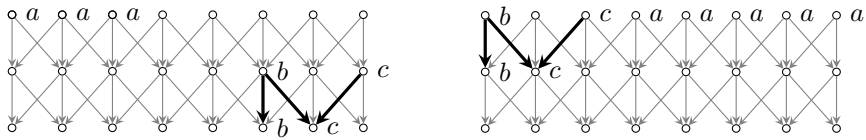
\begin{center} \begin{tikzpicture}[xscale=0.66,yscale=-0.75,>=stealth]
\tikzstyle{v}=[circle,minimum size=1mm,inner sep=0pt,draw] \node[v]
(a1) at (1,1) {}; \node[v] (a2) at (2,1) {}; \node[v] (a3) at (3,1)
{}; \foreach \i in {1,...,8} \foreach \j in {1,...,3} \node[v]
(G-\i-\j) at (\i,\j) {}; \foreach \i in {1,...,8} \foreach \j/\o
in {1/2,2/3} \draw[->,thin,gray] (G-\i-\j) -- (G-\i-\o); \foreach
\i/\n in {1/2,2/3,3/4,4/5,5/6,6/7,7/8} \foreach \j/\o in {1/2,2/3}
{ \draw[->,thin,gray] (G-\i-\j) -- (G-\n-\o); \draw[->,thin,gray]
(G-\n-\j) -- (G-\i-\o); } \node[right=-0pt of a1.east] {$a$};
\node[right=-0pt of a2.east] {$a$}; \node[right=-0pt of a3.east]
{$a$}; \node[right=-0pt of G-6-2.east] {$b$}; \node[right=-0pt
of G-6-3.east] {$b$}; \node[right=-0pt of G-8-2.east] {$c$};
\node[right=-0pt of G-7-3.east] {$c$}; \node at (1,3.15) {};
\draw[->,very thick] (G-6-2) -- (G-6-3); \draw[->,very thick] (G-6-2)
-- (G-7-3); \draw[->,very thick] (G-8-2) -- (G-7-3); \end{tikzpicture}
\hspace*{0.8cm} \begin{tikzpicture}[xscale=0.66,yscale=-0.75,>=stealth]
\tikzstyle{v}=[circle,minimum size=1mm,inner sep=0pt,draw] \foreach
\i in {1,...,8} \foreach \j in {1,...,3} \node[v] (G-\i-\j) at
(\i,\j) {}; \foreach \i in {1,...,8} \foreach \j/\o in {1/2,2/3}
\draw[->,thin,gray] (G-\i-\j) -- (G-\i-\o); \foreach \i/\n in
{1/2,2/3,3/4,4/5,5/6,6/7,7/8} \foreach \j/\o in {1/2,2/3} {
\draw[->,thin,gray] (G-\i-\j) -- (G-\n-\o); \draw[->,thin,gray]
(G-\n-\j) -- (G-\i-\o); } \node[right=-0pt of G-4-1.east] {$a$};
\node[right=-0pt of G-5-1.east] {$a$}; \node[right=-0pt of G-6-1.east]
{$a$}; \node[right=-0pt of G-7-1.east] {$a$}; \node[right=-0pt
of G-8-1.east] {$a$}; \node[right=-0pt of G-1-1.east] {$b$};
\node[right=-0pt of G-1-2.east] {$b$}; \node[right=-0pt of G-3-1.east]
{$c$}; \node[right=-0pt of G-2-2.east] {$c$}; \node at (1,3.15) {};
\draw[->,very thick] (G-1-1) -- (G-1-2); \draw[->,very thick] (G-1-1)
-- (G-2-2); \draw[->,very thick] (G-3-1) -- (G-2-2); \end{tikzpicture}
\end{center}

By Lemma~\ref{connected}, any SP extension $Q'$ has an SP expression
as $\alpha\beta$.  Now by Lemma~\ref{depth}, at least one of $\alpha$
or $\beta$ is not SP.  Moreover, the subnetwork of just the last three
columns is isomorphic to $ns(3,3,3)$, so its activities that are in
either $\alpha$ or $\beta$ must form a non-SP subnetwork.  Without loss
of generality, suppose this is $\beta$ (in the degenerate case there
may then be no activities in $\alpha$ from the last three columns).

Now $a_{3,6},a_{3,7},a_{3,8}$ must all be in $\beta$ by
Lemma~\ref{updownclosed}, by a similar argument to that in the proof
of Lemma~\ref{depth}.  There are now two possibilities.

The first is that at least one of $a_{1,1},a_{1,2},a_{1,3}$ appears
in $\alpha$.  In this case, denote this activity by $a$.  Further, at
least two of $a_{2,6},a_{2,7},a_{2,8}$ must be in $\beta$, and these
two together with two of $a_{3,6},a_{3,7},a_{3,8}$ then forms an \N
subnetwork $Q_\N$ of $Q$.  Note that in $Q$, $a$ does not precede
any of the activities of $Q_\N$.

The second possibility is that $a_{1,1},a_{1,2},a_{1,3}$ are all in
$\beta$.  Then by Lemma~\ref{updownclosed}, $a_{2,1}$ and $a_{2,2}$
are both in $\beta$ as well, when $a_{1,1},a_{1,3},a_{2,1},a_{2,2}$
forms an \N subnetwork $Q_\N$ of $Q$.  In this case, consider the
activities $\{a_{1,4},a_{1,5},a_{1,6},a_{1,7},a_{1,8}\}$.  At least
one of these must be in $\alpha$, by Lemma~\ref{updownclosed} and
since $\alpha$ is non-empty.  Denote this activity by $a$.  In $Q$,
$a$ does not precede any of the activities of $Q_\N$.

In either case, in $Q'$ there must be two activities $b$ and $c$ of
$Q_\N$ which form an antichain in $Q$ but a chain in $Q'$.  Now $a$
and $b$ forms an antichain in $Q$ but $a$ precedes $b$ in $Q'$,
and the same observation holds for $a$ and $c$.  Hence $\{a,b,c\}$
forms an antichain in $Q$ but a chain in $Q'$.

Let $T(a) = T(b) = T(c) = 1$ and $T(x) = \epsilon$ for every other
activity $x$.  The slowdown of $Q'$ is then at least $3/(1+2\epsilon)$,
which can be made arbitrarily close to $3$.  In particular, if
$\epsilon = 1/10$ then the slowdown is at least $5/2$. \end{proof}
We next state a workload-dependent conjecture, and provide evidence
for it.

\section{New conjecture}

\begin{conjecture}\label{factorfourthree} For any activity network
$G=(V,E)$ and workload $t:V \rightarrow (0,\infty)$, there exists an
SP extension $G_{SP}$ of $G$, such that \[ \frac{T(G_{SP},t)}{T(G,t)}
\leq \frac{4}{3}. \] \end{conjecture}

\noindent We now need to consider the evidence to support this
conjecture.  At least four activities are required to represent a
non-SP network, and the only non-SP network on four activities is the
\N network given in Figure~\ref{N}.  We start by proving the result
for the case of four activities.

\begin{theorem} \label{fournode} Let $G^4$ be an activity network with
four activities and workload $t$, then there exists an SP extension
$G^4_{SP}$ of $G^4$ such that $T(G^4_{SP},t)/T(G^4,t) \leq 4/3$.
\end{theorem} \begin{proof} All networks with four activities except
the \N network are SP, for which $T(G^4_{SP},t)/T(G^4,t) = 1 \leq 4/3$.
In the case of $G^4 = \N$, label the activities of \N so that it has
edges $(a,c)$, $(a,d)$ and $(b,d)$. There are then three minimal SP
extensions (in the sense that every other SP extension contains one
of these as a subnetwork):

\noindent \begin{tabular}{ll} \parbox[b]{0.45\textwidth}{
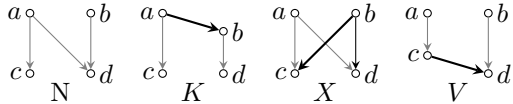
\begin{description} \item[($K$):] $(a,c)$, $(a,d)$, $(b,d)$, $(a,b)$
\item[($X$):] $(a,c)$, $(a,d)$, $(b,d)$, $(b,c)$ \item[($V$):]
$(a,c)$, $(a,d)$, $(b,d)$, $(c,d)$.  \end{description} } &
\hspace*{-0.5cm} \begin{tikzpicture}[yscale=-0.8,xscale=0.8,>=stealth]
\tikzstyle{v}=[circle,minimum size=1mm,inner sep=0pt,draw] \node[v]
(a) at (0,0) {}; \node[v] (b) at (1,0) {}; \node[v] (c) at (0,1) {};
\node[v] (d) at (1,1) {}; \node at (0.5,1.3) {$\N$}; \node[left=-2pt
of a.west] {$a$}; \node[right=-2pt of b.east] {$b$}; \node[left=-2pt
of c.west] {$c$}; \node[right=-2pt of d.east] {$d$}; \draw[->,gray]
(a) to (c); \draw[->,gray] (a) to (d); \draw[->,gray] (b) to (d);
\end{tikzpicture} \begin{tikzpicture}[yscale=-0.8,xscale=0.8,>=stealth]
\tikzstyle{v}=[circle,minimum size=1mm,inner sep=0pt,draw] \node[v]
(a) at (0,0) {}; \node[v] (b) at (1,0.30) {}; \node[v] (c) at (0,1) {};
\node[v] (d) at (1,1) {}; \node at (0.5,1.3) {$K$}; \node[left=-2pt
of a.west] {$a$}; \node[right=-2pt of b.east] {$b$}; \node[left=-2pt
of c.west] {$c$}; \node[right=-2pt of d.east] {$d$}; \draw[->,thick]
(a) to (b); \draw[->,gray] (a) to (c); \draw[->,gray] (b) to (d);
\end{tikzpicture} \begin{tikzpicture}[yscale=-0.8,xscale=0.8,>=stealth]
\tikzstyle{v}=[circle,minimum size=1mm,inner sep=0pt,draw]
\node[v] (a) at (0,0) {}; \node[v] (b) at (1,0) {}; \node[v] (c)
at (0,1) {}; \node[v] (d) at (1,1) {}; \node at (0.5,1.3) {$X$};
\node[left=-2pt of a.west] {$a$}; \node[right=-2pt of b.east]
{$b$}; \node[left=-2pt of c.west] {$c$}; \node[right=-2pt
of d.east] {$d$}; \draw[->,gray] (a) to (c); \draw[->,gray]
(a) to (d); \draw[->,thick] (b) to (c); \draw[->] (b) to (d);
\end{tikzpicture} \begin{tikzpicture}[yscale=-0.8,xscale=0.8,>=stealth]
\tikzstyle{v}=[circle,minimum size=1mm,inner sep=0pt,draw] \node[v]
(a) at (0,0) {}; \node[v] (b) at (1,0) {}; \node[v] (c) at (0,0.7) {};
\node[v] (d) at (1,1) {}; \node at (0.5,1.3) {$V$}; \node[left=-2pt
of a.west] {$a$}; \node[right=-2pt of b.east] {$b$}; \node[left=-2pt
of c.west] {$c$}; \node[right=-2pt of d.east] {$d$}; \draw[->,gray]
(a) to (c); \draw[->,thick] (c) to (d); \draw[->,gray] (b) to (d);
\end{tikzpicture} \end{tabular}

Denote $t(x)$ by $x$ for each $x \in \{ a,b,c,d \}$.  A quantity
such as $3(x+y)$ can be written $xyxyxy$ or using commutativity, just
$xxxyyy$.  Also, if $x \le y$ and $x \le z$ then the conclusion $x \le
\max\{y,z\}$ can instead be written as $x \le y \+ z$.  Now $T(\N)
= \max\{ac, ad, bd\} = ac \+ ad \+ bd$, $T(\K) = \max\{ac, abd\} =
ac \+ abd$, $T(\X) = \max\{ac, ad, bc, bd\} = ac \+ ad \+ bc \+ bd$,
and $T(\V) = \max\{acd, bd\} = acd \+ bd$.

The slowdown is always at least 1, so suppose it is greater than 1
(if it is equal to 1 then the theorem is true).  Then each of $T(\K)$,
$T(\X)$, and $T(\V)$ must exceed $T(\N)$.  Now if $acd \le bd$ then
$T(\V) = bd \le T(\N)$, a contradiction, so $acd > bd$, and hence $ac
> b$.  If $abd \le ac$ then $T(\K) = ac \le T(\N)$, a contradiction,
so $abd > ac$, and hence $bd > c$.  If $b \le a$ then $T(\X) = ac \+
ad \le T(\N)$, a contradiction, so $b > a$.  If $c \le d$ then $T(\X)
= ad \+ bd \le T(\N)$, a contradiction, so $c > d$.  Combined, this
yields $ac > b > a$ and $bd > c > d$.  This leads to $T(\N) = ac \+
bd$, $T(\K) = abd$, $T(\X) = bc$, and $T(\V) = acd$.  Of the three
possibilities for an SP extension with minimal makespan, we analyse \K\
(symmetric to \V); \X\ is similar.

Since $\K$ has minimal makespan among SP extensions, $adb \le bc$
and $adb \le acd$, so $ad \le c$ and $b \le c$.  Hence $abd \le cc$,
so $aaabbbddd \le abccccd$.  Therefore either $bbbddd \le acccc$ or
$aaa \le bd$.  In the first case, $aaabbbddd \le aaaacccc$, and in the
second case, $aaabbbddd \le bbbbdddd$.  In either event, $aaabbbddd
\le aaaacccc \+ bbbbdddd$.  However, $abd$ is just $T(\K)$ and $ac
\+ bd$ is just $T(\N)$, so $T(\K)/T(\N) \le 4/3$.  \end{proof}
Each of the three minimal SP extensions of \N with the workload
$t(a)=1$, $t(b)=2$, $t(c)=2$ and $t(d)=1$ has the same makespan of 4,
while $T(\N,t) = 3$, so the slowdown in this case is at least $4/3$.
This shows that if the workload-dependent bounded slowdown conjecture
holds, then the $4/3$ bound is tight.

Theorem~\ref{fournode} is independent of specific workloads.  This is
also the case for the next theorem.  The five-activity case requires
case analysis but it is done by contradiction rather than by the
direct method used in the four-activity case, and it also uses
max-plus algebra.  Some additional remarks are necessary.

Directed acyclic graphs can be decomposed into modules
\cite{McConnell2005:linear}.  When the edges form a transitive
relation, modules have either series or parallel structure, or
cannot be further decomposed.  Modular decomposition for activity
networks can then be thought of as an extension of the SP grammar
in Section~\ref{structure} by adding a terminal $\mathcal{N}$
representing those networks that cannot be further decomposed in
series or in parallel.  Such indecomposable networks include the \N
network and $ns(d,w,3)$ for $d\ge 3$ and $n\ge 3$.

\begin{theorem} \label{fivenode} Let $G^5$ be an activity network with
five activities and workload $t$ then there exists an SP extension
$G^5_{SP}$ of $G^5$ such that $T(G^5_{SP},t)/T(G^5,t) \leq 4/3$.
\end{theorem} \begin{proof} There are 16 non-isomorphic non-SP
activity networks with five activities.  An activity network and its
dual\footnote{The dual of a directed graph is the graph with its edges
reversed.} have the same slowdown results and we need only consider 9
activity networks.  Six of these can be analysed using decomposition
which yields an SP network with unit slowdown, together with an \N
network to which Theorem~\ref{fournode} can be applied, and the two
slowdowns can then be combined \cite[Theorem 5.12]{Salamon2001:thesis}.
For the three remaining indecomposable networks, the minimal SP
extensions are identified, and each case is checked using arguments
similar to those in the proof of Theorem~\ref{fournode}, yielding
sets of inequalities which each lead to a contradiction if slowdown
greater than 4/3 is assumed. \end{proof}

\section{Programmatic approach}

The six-activity case has been checked using an approach that is
now described.  Our implementation also verified the proofs for 4
and 5 activities.

For a fixed number of activities $n$, we want to consider all non-SP
networks with $n$ activities, and for each of these, to show that for
every possible workload there is a SP extension which achieves the
$4/3$ bound.  Working inductively, for networks with fewer than $n$
activities we have already shown the $4/3$ bound.  We also only need
to consider activity networks up to isomorphism.  Additionally, we do
not need to consider networks that can be decomposed such that there is
at least one series or parallel node in the decomposition, since the
slowdown is then bounded above by the slowdown of an activity network
with less than $n$ activities \cite[Theorem 5.12]{Salamon2001:thesis}.
Therefore we need only consider indecomposable networks and those
which are decomposable but where every module is indecomposable.

The overall schema is to consider each possible activity network
$G$ in turn, assuming that it is a counterexample.  Each of its
SP extensions then has slowdown exceeding $4/3$.  This generates a
system of inequalities, and we can then demonstrate that this system
has no solution.

First all possible $n$-activity networks are generated and
classified into SP, decomposable (but not SP), or indecomposable.
Isomorphic activity networks are discarded, reducing the number of
candidate counterexamples.  For some candidate $G$, consider each
minimal extension $H$.  Only considering minimal extensions is valid
because any non-minimal extension $H'$ will give $T(H) \leq T(H')$.
Every SP extension $H$ exceeds the $4/3$ bound, so we require that
$4T(G) < 3T(H)$ for each such extension.  It is also necessary to
consider some extensions that are not SP as these generate additional,
necessary constraints.  Specifically, we consider the decomposable
extensions because they have slowdown of at most $4/3$.  If $T(G) =
T(H)$ for a decomposable extension $H$ then by the inductive hypothesis
we could find an SP extension of $G$ that would meet the bound, hence
we require that $T(G) < T(H)$ for every decomposable extension $H$
of $G$.

We now need to ask which workloads can allow all these constraints
to hold simultaneously. Since $T(H) = \max\{ T(C) \mid C \text{ is a
chain in\ } H \}$, we can consider each possible chain as a critical
path and generate additional constraints that $T(C) \geq T(D)$ for
all chains $D$ in $H$.  Hence we need to consider the disjunction of
the sets of inequalities \[ \{ T(G) < T(C) \} \cup \{ 4T(G) < 3T(C) \}
\cup \{ T(C) \geq T(D) \mid D \text{ a chain in\ } H, D \neq C \} \]
for every maximal chain $C$ in $H$.  We only need to consider maximal
chains since non-maximal chains have lower makespan.  The makespan of
a chain is simply the sum of its activity durations, so each choice of
critical path $C$ generates a system of linear inequalities expressed
with variables that represent the unknown activity durations.

These inequalities can now be fed to a constraint solver such as
\texttt{clp(q)} \cite{Holzbaur1995:ofai} to check if a workload does
exist that meets the constraints.  If one is found then we have found
a counter-example to the $4/3$ conjecture.  For 4, 5, and 6 activities,
an exhaustive search showed that no counterexamples exist.

We have proved formally that the 4/3 bound holds for the
four-activity and five-activity case, and we have a programmatic
proof of the six-activity case.  This provides some evidence that
Conjecture~\ref{factorfourthree} is true.  The techniques used for
smaller indecomposable networks can be applied to the seven-activity
case also.  However, the systems of inequalities are too large to
handle with the tools currently used, so such a proof would require
new techniques or tools.

\section{Conclusions and further work}

Series-parallelising an activity network is done implicitly when a
program is expressed in an inherently series-parallel formalism, or
explicitly for the purposes of aiding scheduling.  We now consider
the implications of the bound for LC slowdown, the disproof of the
factor of 2 conjecture, and the new factor of 4/3 conjecture.

As shown in Section~\ref{lcbound}, LC slowdown is bounded above.
If all activities have very similar durations, a good bound is obtained
and LC extensions are useful. However, this bound is not necessarily
tight when durations vary.

In the motivating example, deciding which series-parallelisation
to use at the time of writing the program forces a particular
series-parallelisation before the workload is known.  Consider a
parallel programming environment that only allows SP activity
networks to be expressed.  At the time of writing, MATLAB is one
such environment and we believe that in practice both Mathematica and
OpenCL also require activity networks to be SP\footnote{ Mathematica
and OpenCL both provide SP constructs, as well as more general methods
to specify synchronization between activities; unfortunately these
require creating objects for each precedence constraint.  Such a
heavy-weight mechanism only makes sense if activities are all very
large (for instance, if the program consists of just a few threads),
or there are only few precedence constraints.}.

Theorem~\ref{nofactortwo} shows that requiring the
series-parallelisation to be chosen before the workload is known
accurately, may result in slowdown of more than $2$.  Iterating the
construction for larger NS networks (of greater width as well as depth)
allows the slowdown to be forced to be arbitrarily large.

Neighbour-synchronised networks are common in practice and may be
quite large.  The workload in practice may be different to what
was expected when writing the program; for instance, contention for
shared resources, communication delays, and cache misses are just
some of the stochastic effects that affect parallel computation and
that may produce large variations in the duration of an activity.
Therefore, choosing a series-parallelisation without taking into
account possible variations in workload may lead to large slowdown.

If one postpones the decision, it may be possible to do automated
analysis at compile time, or the scheduler may be able to work
around any locally arising bottlenecks due to stochastic variation
in activity durations.  Hence it would seem to be worthwhile allowing
sufficient expressivity in the language so that one can more closely
approximate the activity network of a computation.

On a positive note, if one can find a series-parallelisation that gives
one the conjectured 4/3 bound, then the impact of adding constraints
is limited -- the program will only take one-third as long again
as it would have taken without the additional constraints and this
seems a reasonable penalty to pay to obtain a structure that makes
many scheduling problems easier.

However, one needs to take into account the cost of finding a
series-parallel\-isation that achieves the bound.  Consider the
optimisation problem

\medskip

\prob{MINIMUM SERIES-PARALLELISATION (MSP)}{poset $G$, workload
$t \colon V(G) \rightarrow (0,\infty)$}{poset $H$, $H$ is a SPE of
$G$}{minimise $T(H)$}

\medskip

\noindent Let $|x|$ denote the size of an instance $x$ of MSP.
It is easy to show that MSP is in the complexity class NPO
\cite{Crescenzi1999:structure}.  Computing the level-constrained
extension of an activity network can be done in polynomial time as
discussed in Section~\ref{structure}.  The approximation ratio of
this procedure is bounded by $2^{O(|x|^2)}$.  MSP is therefore in
the class exp-APX, which is strictly contained in NPO unless P $=$
NP \cite{Crescenzi1999:structure}.

Conjecture~\ref{factorfourthree} implies that MSP can
be approximated within a factor of $4/3$, but there is not
necessarily a polynomial-time algorithm that can achieve this.
A branch-and-bound algorithm for solving MSP never needs to consider
more than $2^{O(|x|^2)}$ possible extensions, each corresponding to
a subset of edges.

So a polynomial-time algorithm achieves slowdown of at
most $2^{O(|x|^2)}$.  On the other hand, an SP extension
with minimal slowdown can be found in $2^{O(|x|^2)}$ time, and
Conjecture~\ref{factorfourthree} would bound this slowdown as being
at most $4/3$.  It is not clear how to close this gap; it appears
possible that MSP is exp-APX-hard.

MSP also seems related to the classical decision problem MINIMUM
PRECEDENCE CONSTRAINED SCHEDULING (MPCS) \cite{Garey1979:computers},
which is NP-complete.  The difficulty of MPCS derives from
there being only a limited number of processors.  In contrast,
MSP appears to be difficult because the output network must
be series-parallel.  The $(4/3-\epsilon)$-inapproximability
of MPCS \cite{Graham1979:optimization} suggests that a similar
inapproximability result may exist for MSP.

Several directions for future work are envisaged. The first relates
to the proof of Conjecture~\ref{factorfourthree}, at least for
7 activities. This requires improving the implementation so its
correctness could be verified and finding more powerful techniques
that avoid case analysis.  Second, a programming construct to specify
NS networks could be added to existing programming environments and
its performance established.  Finally, if the decision version of MSP
could be shown to be NP-complete, perhaps by reduction from MPCS, then
the NP-hardness of MSP would follow. Proving that MSP is exp-APX-hard
is another goal.

\paragraph{Acknowledgments} We thank Scott Hazelhurst and Conrad
Mueller for advice and guidance, and Stanislav \v{Z}ivn\'y for
helpful discussions.


\begin{thebibliography}{10}

\bibitem{bisseling2004}
R.~H. Bisseling.
\newblock {}{\em Parallel Scientific Computation: A Structured Approach using
  {BSP} and {MPI}}.
\newblock Oxford University Press, 2004.

\bibitem{Crescenzi1999:structure}
P.~Crescenzi, V.~Kann, R.~Silvestri, and L.~Trevisan.
\newblock Structure in approximation classes.
\newblock {}{\em SIAM Journal on Computing}, {\bf 28}, pp. 1759--1782, 1999.
\newblock \href {http://dx.doi.org/10.1137/S0097539796304220}
  {\path{doi:10.1137/S0097539796304220}}.

\bibitem{cuninghame79}
R.~Cuninghame-Green.
\newblock {}{\em Minimax Algebra}.
\newblock Lecture Notes in Economics and Mathematical Systems~{\bf 166}.
  Springer, 1979.

\bibitem{dodin85}
B.~Dodin.
\newblock \href{http://www.jstor.org/stable/170717}{Bounding the project
  completion time distribution in {PERT} networks}.
\newblock {}{\em Operations Research}, {\bf 33}, pp. 862--881, 1985.

\bibitem{fishburn85}
P.~C. Fishburn.
\newblock {}{\em Interval orders and interval graphs}.
\newblock Wiley, 1985.

\bibitem{Garey1979:computers}
M.~R. Garey and D.~S. Johnson.
\newblock {}{\em Computers and Intractability: A Guide to the Theory of
  NP-Completeness}.
\newblock W. H. Freeman, 1979.

\bibitem{Gonzalez2002:mapping}
A.~Gonz\'alez-Escribano, A.~J. van Gemund, and V.~{Carde\~noso-Payo}.
\newblock \href{http://www.springerlink.com/content/4kk5ymk8593ktn09/}{Mapping
  unstructured applications into nested parallelism}.
\newblock In {}{\em Proc. Int. Conf. on Vector and Parallel Processing (VECPAR
  2002)}, {LNCS}~{\bf 2565}, pp. 407--420. Springer, 2002.
\newblock \href {http://dx.doi.org/10.1007/3-540-36569-9_27}
  {\path{doi:10.1007/3-540-36569-9_27}}.

\bibitem{Graham1979:optimization}
R.~Graham, E.~Lawler, J.~Lenstra, and A.~Rinnooy~Kan.
\newblock
  \href{http://www.math.ucsd.edu/~fan/ron/papers/79_03_scheduling_survey.pdf}{%
Optimization and approximation in deterministic sequencing and scheduling: a
  survey}.
\newblock {}{\em Annals of Discrete Mathematics}, {\bf 5}, pp. 287--326, 1979.

\bibitem{Holzbaur1995:ofai}
C.~Holzbaur.
\newblock
  \href{http://www.ofai.at/cgi-bin/get-tr?download=1\&paper=oefai-tr-95-09.pdf%
}{{}{\em {OFAI clp(Q,R) Manual}, Edition 1.3.3}}.
\newblock Technical Report TR-95-09, Austrian Research Institute for Artificial
  Intelligence, Vienna, 1995.

\bibitem{kleinoder82}
W.~Klein{\"o}der.
\newblock {}{\em Stochastische Bewertung von Aufgabenstrukturen f{\"u}r
  Hierarchische Mehrrechnersysteme}.
\newblock Technical Report Band 15, Nummer 10, Friedrich Alexander
  Universit{\"a}t Erlangen-N{\"u}rnberg, Institut f{\"u}r Mathematische
  Maschinen und Datenverarbeitung (IMMD), 1982.

\bibitem{Kohler1975:preliminary}
W.~H. Kohler.
\newblock A preliminary evaluation of the critical path method for scheduling
  tasks on multiprocessor systems.
\newblock {}{\em IEEE Transactions on Computers}, {\bf 24}, pp. 1235--1238,
  1975.
\newblock \href {http://dx.doi.org/10.1109/T-C.1975.224171}
  {\path{doi:10.1109/T-C.1975.224171}}.

\bibitem{Kwok1999:static}
Y.-K. Kwok and I.~Ahmad.
\newblock Static scheduling algorithms for allocating directed task graphs to
  multiprocessors.
\newblock {}{\em ACM Computing Surveys}, {\bf 31}, pp. 406--471, 1999.
\newblock \href {http://dx.doi.org/10.1145/344588.344618}
  {\path{doi:10.1145/344588.344618}}.

\bibitem{malony:94}
A.~D. Malony, V.~Mertsiotakis, and A.~Quick.
\newblock Automatic scalability analysis of parallel programs based on
  modelling techniques.
\newblock In {}{\em Computer Performance Evaluation: Modelling Techniques and
  Tools. 7th International Conference, Vienna, Austria, May 3--6, 1994.
  Proceedings}, {LNCS}~{\bf 794}, pp. 139--158. Springer, 1994.
\newblock \href {http://dx.doi.org/10.1007/3-540-58021-2_8}
  {\path{doi:10.1007/3-540-58021-2_8}}.

\bibitem{McConnell2005:linear}
R.~M. McConnell and F.~de~Montgolfier.
\newblock Linear-time modular decomposition of directed graphs.
\newblock {}{\em Discrete Applied Mathematics}, {\bf 145}, pp. 198--209, 2005.
\newblock \href {http://dx.doi.org/10.1016/j.dam.2004.02.017}
  {\path{doi:10.1016/j.dam.2004.02.017}}.

\bibitem{Munshi2009:opencl}
{Munshi, A. (editor)}.
\newblock \href{http://www.khronos.org/registry/cl/}{{}{\em The {OpenCL}
  Specification (Document Revision 33)}}.
\newblock Khronos {OpenCL} Working Group, 2009.

\bibitem{Salamon2001:thesis}
A.~Z. Salamon.
\newblock
  \href{ftp://ftp.cs.wits.ac.za/pub/research/reports/TR-Wits-CS-2001-0.pdf}{{}%
{\em Task Graph Performance Bounds Through Comparison Methods}}.
\newblock Master's thesis, University of the Witwatersrand, Johannesburg, 2001.

\bibitem{Valdes1982:recognition}
J.~Valdes, R.~E. Tarjan, and E.~L. Lawler.
\newblock The recognition of series parallel digraphs.
\newblock {}{\em SIAM Journal on Computing}, {\bf 11}, pp. 298--313, 1982.
\newblock \href {http://dx.doi.org/10.1145/800135.804393}
  {\path{doi:10.1145/800135.804393}}.

\bibitem{valiant90}
L.~G. Valiant.
\newblock A bridging model for parallel computation.
\newblock {}{\em Communications of the ACM}, {\bf 33}, pp. 103--111, 1990.
\newblock \href {http://dx.doi.org/10.1145/79173.79181}
  {\path{doi:10.1145/79173.79181}}.

\bibitem{Vangemund1997:importance}
A.~J.~C. van Gemund.
\newblock The importance of synchronization structure in parallel program
  optimization.
\newblock In {}{\em {ICS} '97: Proceedings of the 11th international conference
  on Supercomputing}, pp. 164--171. ACM, 1997.
\newblock \href {http://dx.doi.org/10.1145/263580.263625}
  {\path{doi:10.1145/263580.263625}}.

\end{thebibliography}
\end{document}